\documentclass[journal]{IEEEtran}

\ifCLASSINFOpdf
\else
   \usepackage[dvips]{graphicx}
\fi
\usepackage{url}

\usepackage{graphicx}

\usepackage{amsmath, amssymb, amsfonts, mathrsfs, amsthm}
\usepackage{xcolor}
\usepackage{dblfloatfix}
\DeclareMathOperator*{\argmax}{arg\,max}
\usepackage[caption=false,font=footnotesize]{subfig}
\newtheorem{proposition}{Proposition}
\usepackage{comment}
\usepackage{braket}  

\begin{document}

\title{Anytime-Valid Quantum State Tomography via Confidence Sequences}

\author{Aldo Cumitini, Luca Barletta, and  Osvaldo Simeone

\thanks{Aldo Cumitini and Luca Barletta are with the Dipartimento di Elettronica,
Informazione e Bioingegneria, Politecnico di Milano, Milano, Italy. (email:
aldo.cumitini@mail.polimi.it, luca.barletta@polimi.it). O. Simeone is with the Institute for Intelligent Networked Systems (INSI),
Northeastern University London, One Portsoken Street, London E1 8PH,
United Kingdom (email: o.simeone@northeastern.edu). The work of A. Cumitini was supported by the Ermenegildo Zegna Founder's Scholarship 2025-26.  The work of O. Simeone was supported by the European Research Council (ERC) under the European Union’s Horizon Europe Programme (grant agreement No. 101198347), by an Open Fellowship of the EPSRC (EP/W024101/1), and by the EPSRC project (EP/X011852/1). }}

\maketitle

\begin{abstract}
In this letter, we address the problem of developing quantum state tomography (QST) methods that remain valid at any time  during a sequence of measurements. Specifically, the aim is to provide a rigorous quantification of the uncertainty associated with the current state estimate as data are acquired incrementally. To this end, the proposed framework augments existing QST techniques by associating current point estimates of the state with confidence sets that are guaranteed to contain the true quantum state with a user-defined probability. The methodology is grounded in recent statistical advances in anytime-valid confidence sequences. Numerical results confirm the theoretical coverage properties of the proposed anytime-valid QST.
\end{abstract}


\IEEEpeerreviewmaketitle

\section{Introduction}
Quantum state tomography (QST) is a fundamental procedure for characterizing quantum systems, enabling the estimation of an unknown quantum state from measurement outcomes~\cite{yu2022statistical}. Traditional QST methods provide point estimates or credible regions based on accumulated data, but they lack rigorous guarantees in the presence of misspecified models or when inference is performed sequentially. As illustrated in Fig. \ref{fig:bloch_spheres}, the goal of this paper is to augment existing QST techniques by associating current state estimates with confidence sets that are guaranteed to contain the true quantum state with a user-defined probability \emph{uniformly} over time. 

\begin{figure}[!t]
    \centering
    \captionsetup[subfloat]{labelformat=empty,font=normalsize,skip=0pt}

    \subfloat[\normalsize $t = 5$]{%
        \includegraphics[
        width=0.4\linewidth,
        trim = 8.5cm 8.5cm 8.5cm 8.5cm,
        clip
        ]{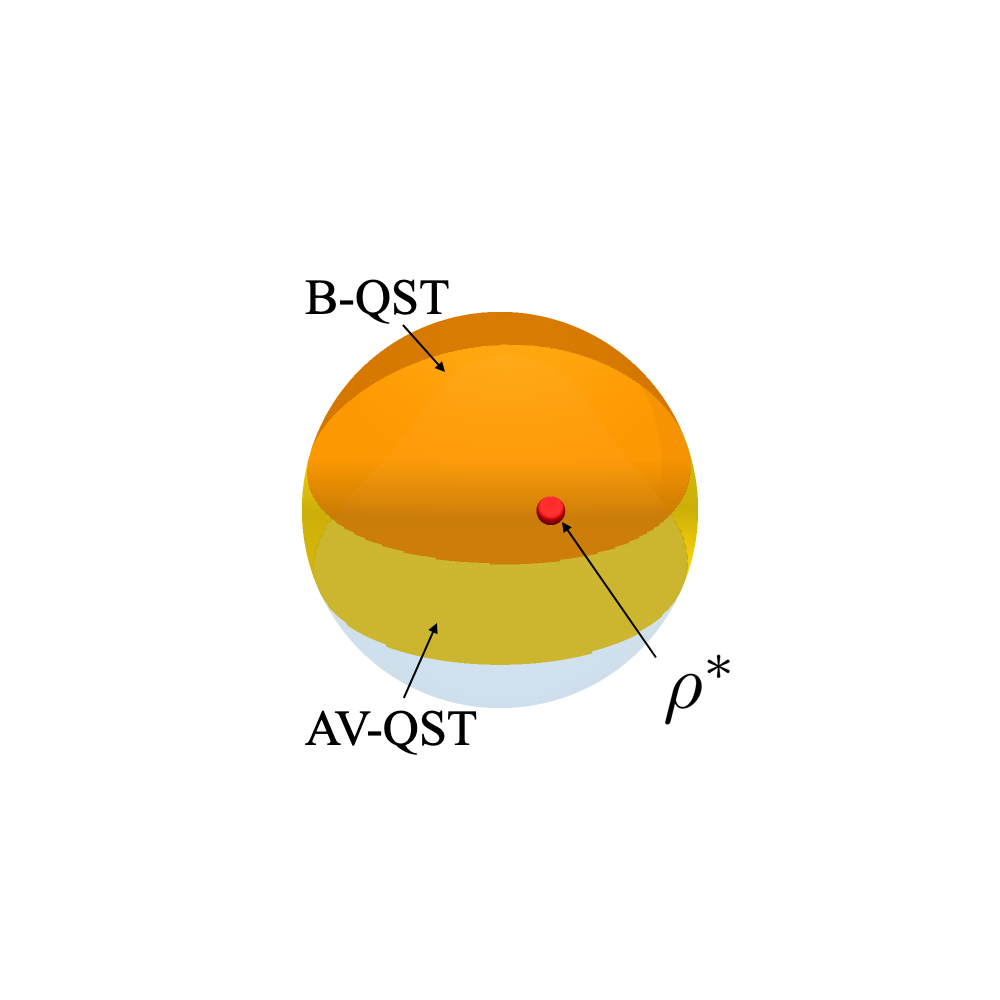}%
        \label{fig:t6}}
    \hfil
    \subfloat[\normalsize $t = 22$]{%
        \includegraphics[
        width=0.4\linewidth,
        trim = 8.5cm 8.5cm 8.5cm 8.5cm,
        clip]{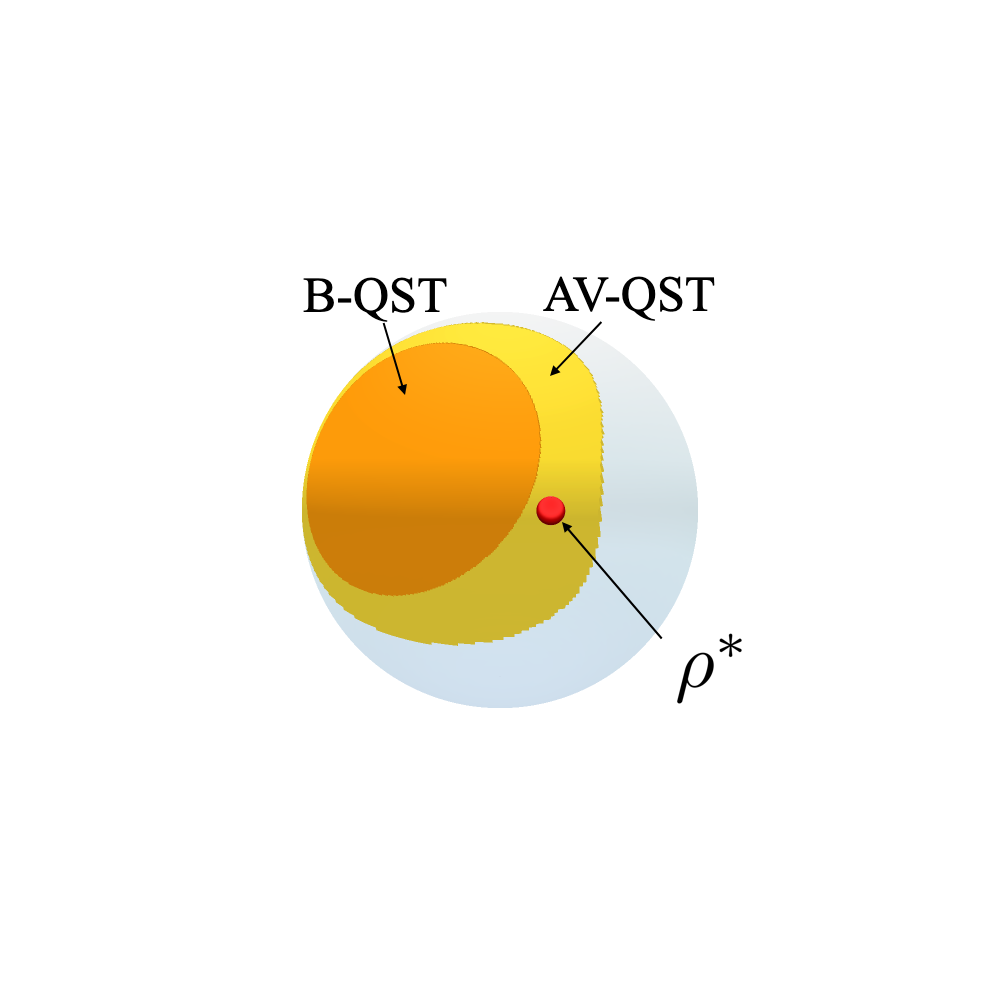}%
        \label{fig:t22}}
    \hfil
    \subfloat[\normalsize $t = 50$]{%
        \includegraphics[
        width=0.4\linewidth,
        trim = 8.5cm 8.5cm 8.5cm 8.5cm,
        clip]{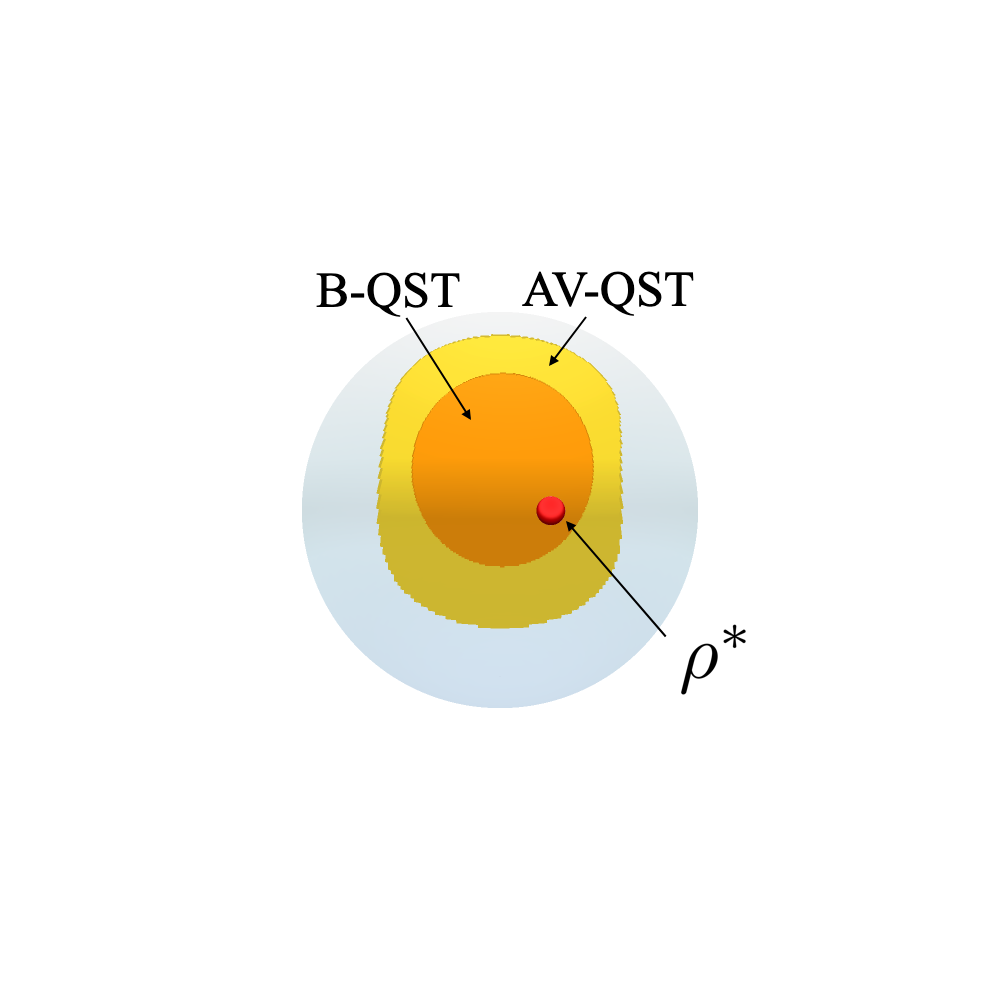}%
        \label{fig:t50}}
    \hfil
    \subfloat[\normalsize $t = 100$]{%
        \includegraphics[
        width=0.4\linewidth,
        trim = 8.5cm 8.5cm 8.5cm 8.5cm,
        clip]{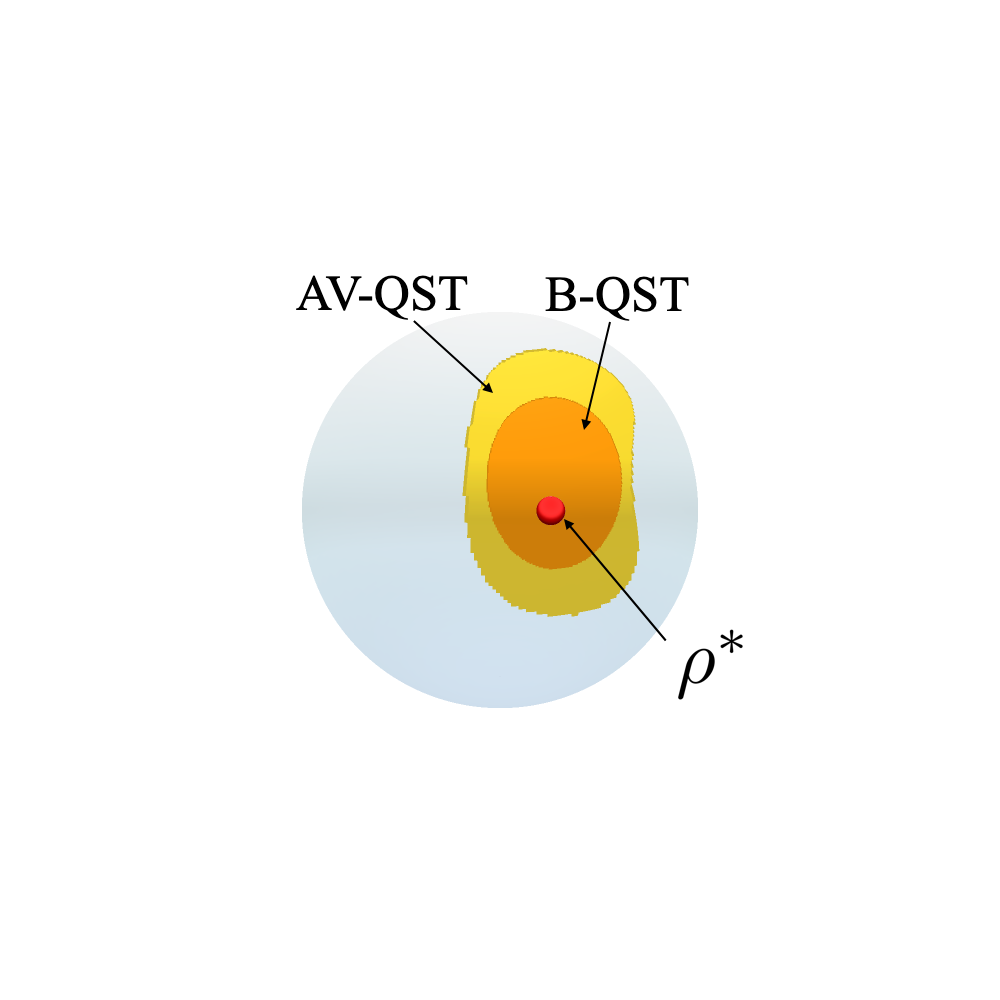}%
        \label{fig:t100}}
    \caption{Uncertainty regions obtained as measurements are carried out on copies of a true state $\rho^*$ over discrete time $t=1,2,...$ using  Bayesian credible regions (B-QST) \cite{Blume_Kohout_2010} and the proposed anytime-valid QST (AV-QST). Bayesian QST (B-QST) only provides asymptotic  statistical guarantee (highlighted at $t=22$); while the proposed AV-QST ensures that the confidence region include the true state $\rho^*$ (red dot) with probability no smaller than $1-\alpha$ \emph{uniformly} over \emph{all} time steps $t$. }
    \label{fig:bloch_spheres}
\end{figure}

QST has been extensively studied from both frequentist and Bayesian perspectives. Maximum likelihood estimation (MLE) for quantum states was explored  in~\cite{PhysRevA.55.R1561}, providing consistent point estimates but no uncertainty quantification. The works~\cite{Blume_Kohout_2010} and \cite{huszar2012adaptive} developed Bayesian methods  to construct posterior estimates and credible regions. The statistical validity of these techniques relies on the correct specification of the prior distribution. Recent advances have explored compressed sensing techniques~\cite{gross2010quantum}, adaptive measurement strategies~\cite{huszar2012adaptive}, and machine learning approaches~\cite{torlai2018neural,carrasquilla2019reconstructing,10.1109/TSP.2025.3546655}. There is also an important line of work on estimating specific properties of quantum states, rather than recovering the full state~\cite{Elben_2022}.

The problem of providing rigorous error bounds for QST has been addressed by several recent works. Notably, the paper~\cite{blume2012robust} introduced confidence regions based on the likelihood ratio that offer  rigorous frequentist coverage guarantees (see also  \cite{Renner} for a related theoretical study). While this method  provides reliable error bounds, it is designed for a \emph{fixed sample size}. That is,  the confidence regions are theoretically  valid only for a predetermined number of measurements decided before data collection begins. This limitation poses challenges in sequential or adaptive scenarios where the experimenter may wish to monitor the state estimate in real-time, make decisions based on intermediate results, or stop the experiment when a desired accuracy is reached.

In contrast, this letter aims at developing  \emph{anytime-valid} confidence sets that maintain their coverage guarantees uniformly over all time steps, including at data-dependent stopping times. This property is critical for sequential decision-making. The proposed method, \emph{anytime-valid quantum state tomography} (AV-QST), augments existing QST methods with rigorous, time-uniform confidence guarantees~\cite{howard2021time}. 


\section{Problem Definition}
\label{sec:problem}

QST aims to reconstruct an unknown quantum state 
$\rho^* \in \mathcal{D}(\mathcal{H})$, where $\mathcal{D}(\mathcal{H})$ denotes the set of 
density matrices acting on a Hilbert space $\mathcal{H}$. 
For a $D$-dimensional quantum system, where $D = 2^m$ for an $m$-qubit system, 
the state $\rho$ is represented by a $D \times D$ complex-valued density matrix. 
 By definition, this matrix is positive semidefinite and has unit trace, i.e.,  $\rho \succeq 0$ and $\mathrm{tr}(\rho)=1$.

In a tomographic experiment, information about the true state $\rho^*$ is inferred from a 
sequence of measurement outcomes obtained from identically prepared copies of the system. 
Each measurement is described by a {positive operator-valued measure} (POVM) $\mathcal{M} = \{ \Pi_x\}_{x = 1}^{M}$, 
where each element $\Pi_x$ is a positive semidefinite matrix and  the equality $\sum_{x = 1}^{M} \Pi_x = I$ is satisfied, with $I$ being the $D\times D$ identity matrix. Accordingly, the probability of observing outcome $x$ when the system is in state $\rho$ is given by Born's rule:
\begin{equation}
\mathbb{P}(x|\rho; \mathcal{M}) = \mathrm{tr}(\Pi_x \rho).
\end{equation}

At each discrete time step $t = 1,2,\ldots,$ a POVM $\mathcal{M}_t=\{ \Pi_{t,x}\}_{x = 1}^{M}$ 
is chosen, and an observation $X_t \sim \mathbb{P}(x_t|\rho; \mathcal{M}_t)$ is collected.
Given the current sequence of measurement outcomes $\{ X_{t^{'}} \}_{t^{'}=1}^t$, 
the goal of QST is to maintain an estimate of the underlying true state $\rho^*$.

Standard frequentist methods \cite{PhysRevA.55.R1561} provide point estimates  $\hat{\rho}_t$. In contrast,  Bayesian methods not only produce a point estimate $\hat{\rho}_t$, but they also aim to  quantify the uncertainty associated with this estimate via a credible set $\mathcal{C}_t \subseteq \mathcal{D}(\mathcal{H})$ (see Fig.~\ref{fig:bloch_spheres}). Specifically, the credible sets should ideally cover the true state $\rho^*$ with some user-defined probability.  However, this statistical validity property hinges on the correct specification of prior distribution which can be practically difficult~\cite{Blume_Kohout_2010}.

Thus motivated, in this paper, we study the problem of complementing frequentist or Bayesian point estimates  $\hat{\rho}_t$ by constructing a sequence of {uncertainty sets} 
$\{ \mathcal{C}_t \}_{t \ge 1}$, with $\mathcal{C}_t \subseteq \mathcal{D}(\mathcal{H})$ that are \textit{anytime-valid}~\cite{howard2021time,Ramdas_2025}. That is, at any point in the measurement sequence, including at random stopping times, 
the reported uncertainty set $\mathcal{C}_t$ maintains a statistically rigorous confidence level $1 - \alpha$ 
for containing the true quantum state $\rho^*$, where $1-\alpha$ is a user-defined probability.  Mathematically, this requirement is expressed by the inequality 
\begin{equation}
    \mathbb{P}\!\left( 
    \forall t \ge 1,\, \rho^* \in \mathcal{C}_t
\right) \ge 1 - \alpha.
\label{eq:prob}
\end{equation} Note that this is a frequentist guarantee, in the sense that the probability is evaluated with respect to the distribution of the observations for any fixed, unknown, true state $\rho^*$.

\section{Conventional Methods}
\label{sec:conventional}

In this section, we summarize baseline frequentist and Bayesian QST methods.

\subsection{Frequentist Quantum State Tomography}

The standard frequentist methodology is based on  \emph{maximum likelihood estimation} (MLE) \cite{PhysRevA.55.R1561}. Given the available measurements $\{ X_{t^{'}} \}_{t^{'}=1}^t$, the MLE  maximizes the probability of the observed measurement outcomes, producing a point estimate of the true state. 

The probability of observing the data under a candidate state $\rho$ is \begin{equation}
\mathcal{L}_t(\rho) = \prod_{t^{'}=1}^{t}\mathbb{P}(X_{t'}| \rho; \mathcal{M}_{t'}) = \prod_{t'=1}^{t} \mathrm{tr}(\Pi_{X_{t'}} \rho),
\label{eq:likelihood}
\end{equation}where we have denoted for simplicity $\Pi_{t,X_{t}}=\Pi_{X_{t}}$. The MLE estimate $\hat{\rho}^{\text{MLE}}$ is then evaluated as the maximizer
\begin{equation}
\hat{\rho}^{\text{MLE}}_t = \argmax_{\rho \succeq 0,\, \mathrm{tr}(\rho)=1} \left\{ \log \mathcal{L}_t(\rho) = \sum_{t'=1}^{t} \log \mathrm{tr}(\Pi_{X_{t'}} \rho )\right\},
\label{eq:MLE_LOG}
\end{equation} where $\log \mathcal{L}_t(\rho)$ is the log-likelihood function. 
%

\subsection{Bayesian Quantum State Tomography}\label{sec:BQST}

Like frequentist QST, Bayesian QST produces a point estimate of the underlying quantum state, but it  also  quantifies the uncertainty of this estimate via  a credible set \cite{Blume_Kohout_2010}. To start, Bayesian QST requires the specification of a prior distribution $\pi_0(\rho)$ over quantum states, such as the Haar measure. 

Given the available data $\{X_{t^{'}}\}_{t^{'}=1}^t$, at time $t$,  the posterior distribution  $\pi_{t}(\rho)$  can be evaluated  recursively by updating the previous posterior distribution $\pi_{t-1}(\rho)$ with the new observation $X_t$ and the corresponding measurement probability $\mathbb{P}(X_t|\rho; \mathcal{M}_t)=\mathrm{tr}(\Pi_{X_t} \rho)$  as \begin{equation}
\pi_t(\rho) \propto \pi_{t-1}(\rho)\,\mathrm{tr}(\Pi_{X_t} \rho).
\label{eq:update_rule}
\end{equation}

\indent A point estimator of the state is typically defined in \emph{Bayesian QST} (B-QST) via the \emph{posterior mean} 
\begin{equation}
    \hat{\rho}^{\text{B-QST}}_t = \int \rho\, \pi_t(\rho)\, d\rho,
    \label{eq:bme}
\end{equation} 
which can be approximated via Monte Carlo methods \cite{Blume_Kohout_2010, huszar2012adaptive}.

Moreover, B-QST augments the point estimate \eqref{eq:bme} with a \emph{credible set} $\mathcal{C}_t$ \cite[Sec. 2.3.2]{Blume_Kohout_2010} targeting a coverage level $1-\alpha$ separately for each time $t$, i.e.,  $\mathbb{P}\!\left( 
     \rho^* \in \mathcal{C}_t
\right|\{X_{t^{'}}\}_{t^{'}=1}^t) \ge 1 - \alpha$. Note that this is a Bayesian, not frequentist, guarantee, in the sense that the probability is evaluated for a fixed dataset $\{X_{t^{'}}\}_{t^{'}=1}^t$ under the given prior distribution for the random set $\rho^*\sim \pi_0(\rho)$. In practice,  the credible set is obtained based on Gaussian   approximations of the posterior distribution \cite{Blume_Kohout_2010} or Monte Carlo methods \cite{Blume_Kohout_2010, huszar2012adaptive}.

\section{Anytime-Valid Quantum Tomography}
\label{sec:avqst}

B-QST  generally fails to satisfy an anytime-validity requirement akin to (\ref{eq:prob}) for a number of reasons: (\emph{i}) the prior distribution may be misspecified;  (\emph{ii}) the posterior distribution is only asymptotically Gaussian; and (\emph{iii}) even in the ideal case of correct prior specification and exact posterior calculation, guarantees are only provided \emph{separately}  for each time $t$. In this section we introduce anytime-valid QST (AV-QST), a framework for online uncertainty quantification in QST that augments any point predictor of the quantum state with anytime-valid confidence sets satisfying the condition (\ref{eq:prob}).  The approach builds on the statistical methodology of confidence sequences \cite{howard2021time}, which is adapted to leverage Born's rule and to operate in the space of density matrices $\mathcal{D}(\mathcal{H})$.

\subsection{Confidence Sequences}

At each time step $t$, AV-QST assumes access to \emph{any} point predictor $\hat{\rho}_{t-1}$, e.g., the MLE in~\eqref{eq:MLE_LOG} or the posterior mean in~\eqref{eq:bme}. Based on this estimate, given the available data $\{X_{t^{'}}\}_{t^{'} = 1}^t$, AV-QST constructs the \emph{likelihood ratio}
\begin{equation}    
{R}_t(\rho)
= \prod_{t^{'}=1}^t 
  \frac{ \mathrm{tr}(\Pi_{X_{t'}} \hat{\rho}_{t'-1})  }
       { \mathrm{tr}(\Pi_{X_{t'}} {\rho})   }.
\label{eq:Martingale_R}
\end{equation}
The likelihood ratio ${R}_t(\rho)$ in \eqref{eq:Martingale_R} quantifies the available evidence against the candidate state $\rho$ relative to the estimates $\hat{\rho}_{t'-1}$ evaluated at each time $t'=1,\dots,t$. Specifically, a lower value of the ratio ${R}_t(\rho)$ indicates that state $\rho$ is more likely to represent the true state $\rho^*$ as compared to these estimates.
Accordingly, a confidence set can be obtained by collecting all quantum states whose likelihood ratio is sufficiently low, i.e., all states for which there is sufficient evidence in their favor as compared to the current best estimate.
This set is evaluated as
\begin{equation}\label{eq:Ct}
    \mathcal{C}_t^{\text{AV-QST}}(\alpha)
= \left\{ \rho \in \mathcal{D}(\mathcal{H}) : {R}_t(\rho) \le 1/\alpha \right\},
\end{equation}
where the threshold $1/\alpha$ will be seen in Proposition 1 below to guarantee the desired coverage condition \eqref{eq:prob}.

The likelihood ratio (\ref{eq:Martingale_R}), instantiated with the MLE as the point estimate, was also used in \cite{blume2012robust} to construct a confidence set $\mathcal{C}_T(\alpha)$ for a \emph{batch} setting with a fixed sample size $T$. The scheme, referred to as \emph{likelihood ratio-based QST} (LR-QST),  returns the single set \begin{equation}\label{eq:LR}\mathcal{C}^{\text{LR-QST}}_T(\alpha)
= \{ \rho \in \mathcal{D}(\mathcal{H}) : R^{\text{LR-QST}}_T(\rho) \le \lambda_\alpha\}\end{equation} at the last time step, where 
\begin{equation} R^{\text{LR-QST}}_T(\rho)=\frac{\mathcal{L}_T(\hat{\rho}^\text{MLE}_T)}{\mathcal{L}_T(\rho)}
\label{eq:r_lrqst}
\end{equation} 
is a likelihood ratio, and the threshold  $\lambda_\alpha$ is obtained by solving numerically an implicit equation (see \cite[Equation (14)]{blume2012robust}). LR-QST was derived with the aim of offering only coverage guarantees for the last time step $T$, and it does not allow for the use of more powerful point estimates such as the posterior mean (\ref{eq:bme}).

From a computational standpoint, LR-QST is inherently tied to the MLE through the likelihood ratio in \eqref{eq:r_lrqst}, requiring a complexity that scales exponentially with the number of qubits \cite{gross2010quantum}. In contrast, AV-QST can be combined with an arbitrary sequence of point estimators, including computationally efficient estimators \cite{article, cai2025onlinequantumstatetomography}. Furthermore, unlike LR-QST, AV-QST employs the fixed threshold $1/\alpha$, and does not require the additional numerical calibration of the threshold $\lambda_\alpha$. Finally, both AV-QST and LR-QST can be substantially less computationally demanding than Bayesian QST, which generally requires Monte Carlo updates over multiple particles to approximate posterior distributions and credible regions.

\subsection{Theoretical Properties}

As stated in the following proposition, which follows directly from  the properties of confidence sequences \cite{howard2021time}, the sequence of sets  $\{\mathcal{C}_t^{\text{AV-QST}}(\alpha)\}_t$ produced by AV-QST satisfies the desired anytime validity property (\ref{eq:prob}).

\begin{proposition}
For any sequence of point estimates $\{\hat{\rho}_t\}_t$, the sequence of AV-QST uncertainty sets  $\{\mathcal{C}_t^{\text{AV-QST}}(\alpha)\}_t$ in (\ref{eq:Ct}) satisfies the anytime-valid coverage  property (\ref{eq:prob}).
\end{proposition}
\begin{proof}
Given the definition of the AV-QST set (\ref{eq:Ct}), the target coverage probability (\ref{eq:prob}) can be written as \begin{equation}
    \mathbb{P}\!( 
    \forall t \ge 1,\, \rho^* \in \mathcal{C}_t^{\text{AV-QST}}(\alpha)
) = \mathbb{P}\!\left( 
    \forall t \ge 1,\, {R}_t(\rho^\ast) \le 1/\alpha
\right),
\label{eq:prob1}
\end{equation}  
where the probability is evaluated with respect to the  distribution of the measurements, i.e.,   $\prod_{t^{'}=1}^{t}\mathbb{P}(x_{t'} |\rho^*; \mathcal{M}_{t'})$.  Furthermore, the likelihood ratio ${R}_t(\rho^\ast)$ can be readily verified to be a test martingale, i.e., a non-negative supermartingale, under the measurement distribution (see Appendix). Ville's inequality is a time-uiform version of Markov's inequality for test martingales \cite{Ramdas_2025}. When applied to the likelihood ratio ${R}_t(\rho^*)$, this inequality directly implies the condition
\begin{equation}
    \mathbb{P}\!\left( 
    \forall t \ge 1,\, {R}_t(\rho^\ast) \le 1/\alpha
\right)\geq 1-\alpha.
\label{eq:prob2}
\end{equation} Combining \eqref{eq:prob2} with (\ref{eq:prob1}) completes the proof.    
\end{proof}

\section{Numerical Results}
\label{sec:results}

\begin{figure*}[!t]
    \centering
    \subfloat[]{%
        \includegraphics[width=0.33\linewidth]{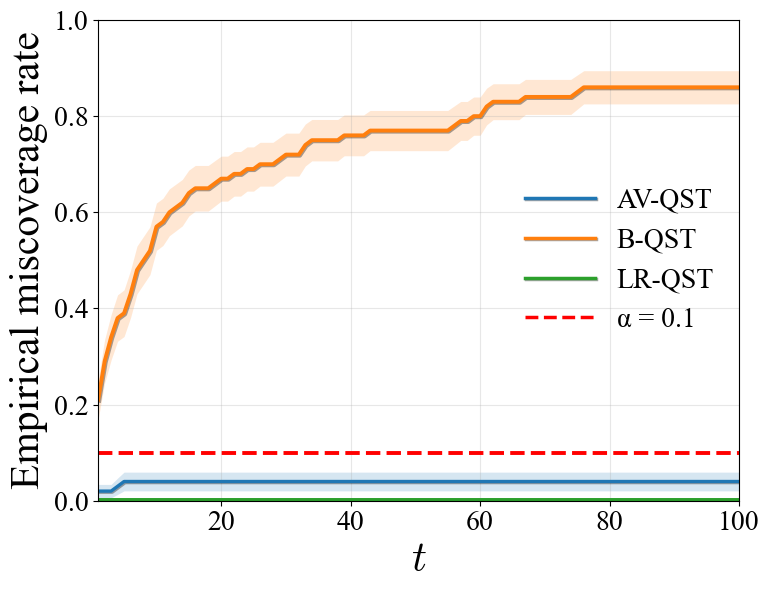}%
        \label{fig:a1qub}}
    \hfil
    \subfloat[]{%
        \includegraphics[width=0.33\linewidth]{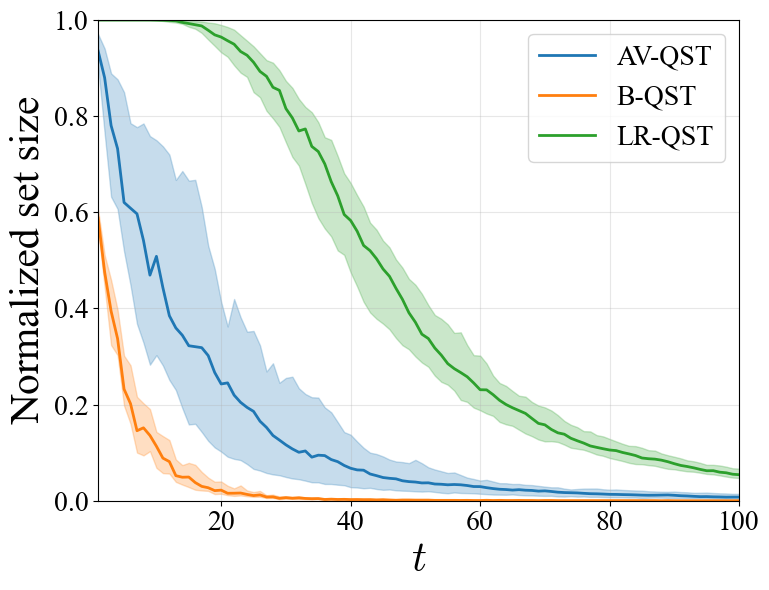}%
        \label{fig:b1qub}}
    \hfil
    \subfloat[]{%
        \includegraphics[width=0.33\linewidth]{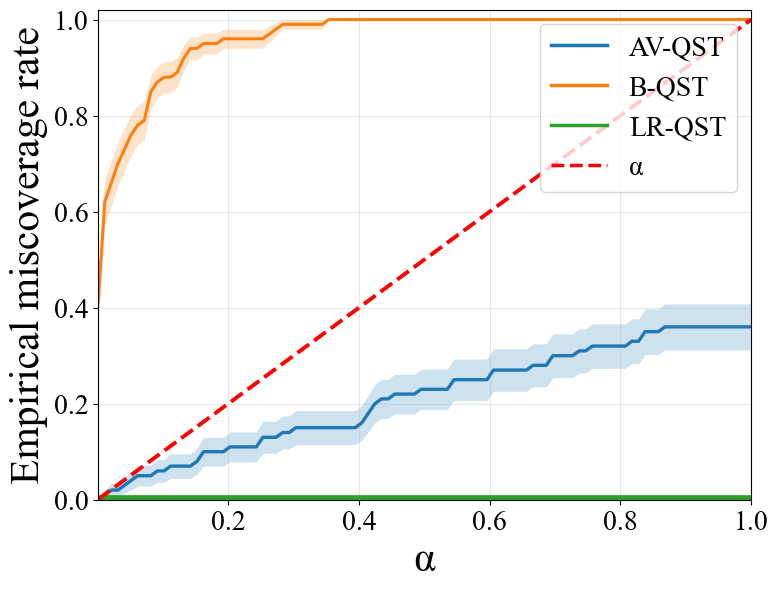}%
        \label{fig:c1qub}}
    \caption{ Empirical miscoverage and normalized set size for B-QST \cite{Blume_Kohout_2010}, LR-QST \cite{blume2012robust}, and AV-QST (this paper) for the two-qubit setting. The lines represent median values and the shaded areas correspond to the interquartile range (25th-75th percentiles) over 100 runs.}
    \label{fig:2qub_results}
\end{figure*}

\begin{figure*}[!t]
    \centering
    \subfloat[]{%
        \includegraphics[width=0.33\linewidth]{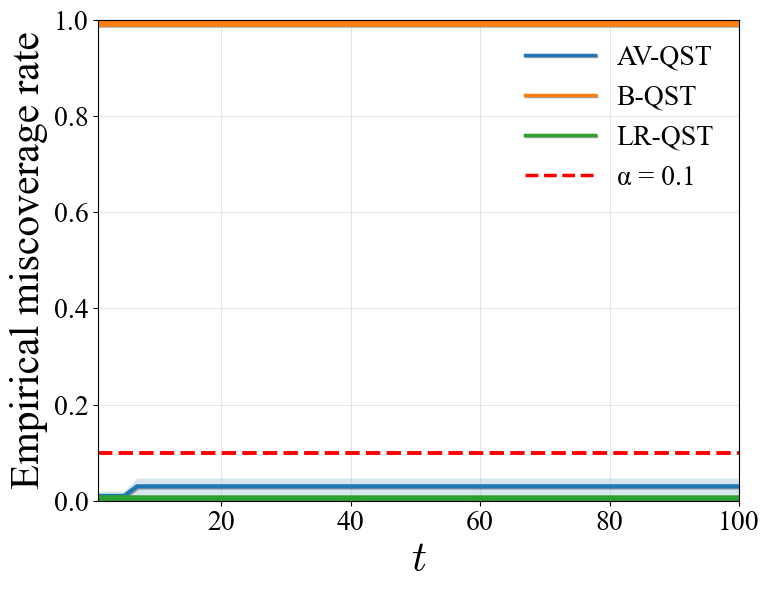}%
        \label{fig:a2qub}}
    \hfil
    \subfloat[]{%
        \includegraphics[width=0.33\linewidth]{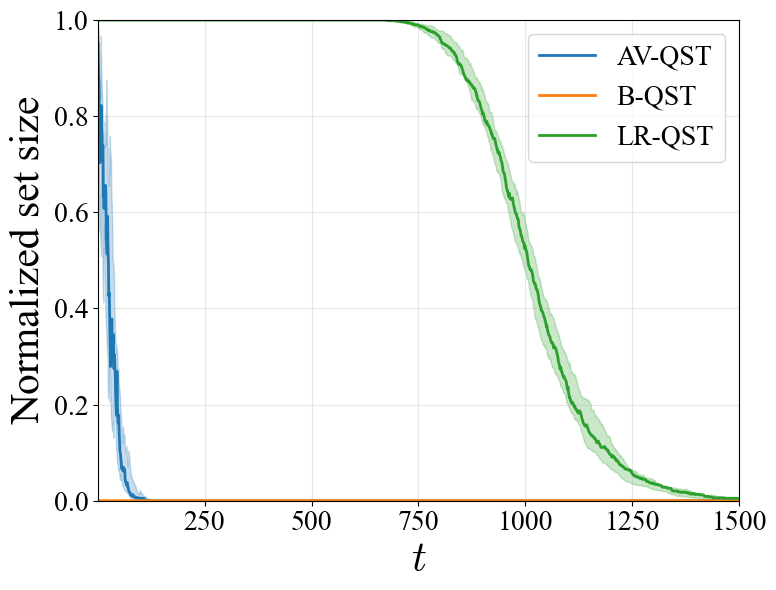}%
        \label{fig:b2qub}}
    \hfil
    \subfloat[]{%
        \includegraphics[width=0.33\linewidth]{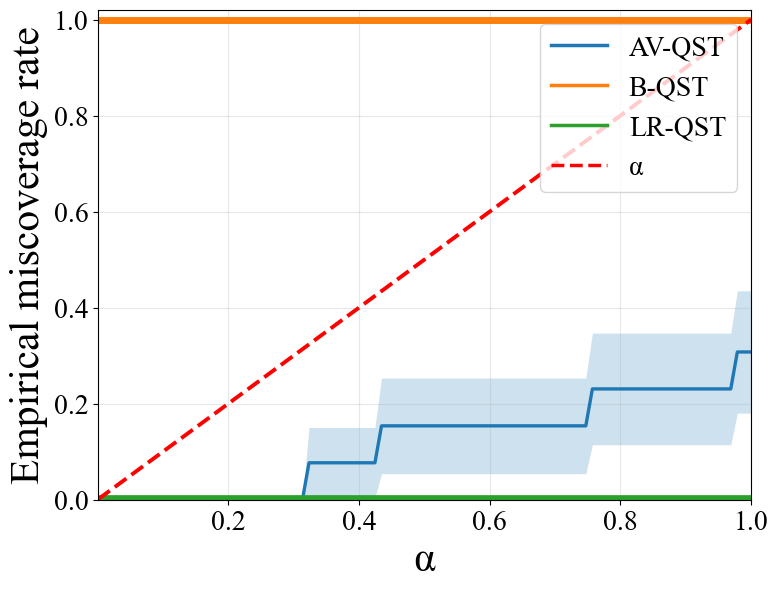}%
        \label{fig:c2qub}}
    \caption{Empirical miscoverage and normalized set size for B-QST \cite{Blume_Kohout_2010}, LR-QST \cite{blume2012robust}, and AV-QST (this paper) for the four-qubit setting. The lines represent median values and the shaded areas correspond  to the interquartile range (25th-75th percentiles) over 15 runs.}
    \label{fig:4qub_results}
\end{figure*}

In this section, we present  numerical results for simulated tomography experiments. All simulations were implemented in Python and executed on a M3 Max Macbook Pro. Throughout, we perform local measurements using the minimal informationally complete POVM (MIC-POVM) separately to the two qubits (see \cite[Eqs. (22), (23)]{Renes_2004}). For each simulated Monte Carlo run, we generate a random pure qubit state $\rho^*$  following the Haar uniform distribution.  LR-QST is run by evaluating the set (\ref{eq:LR}) separately for each time $t=T$. The point estimate $\hat{\rho}_t$ used in AV-QST is the MLE.

 Fig. \ref{fig:bloch_spheres} displays the confidence regions produced by the proposed AV-QST, along with B-QST (Sec. \ref{sec:BQST}) for a given run. The figure shows that the true state $\rho^*$ is not included  in the B-QST region at time $t=22$, reflecting the theoretical property that B-QST does not satisfy the anytime-validity requirement, contrary to the AV-QST region.
 
Fig.~\ref{fig:2qub_results} reports the  results averaged over multiple runs  for the two-qubit case ($D=4$), and Fig.~\ref{fig:4qub_results} for the four-qubit case ($D=16$). For the target miscoverage rate $\alpha = 0.1$, we plot the empirical miscoverage rate and the normalized set size as a function of time $t$ in panels (a) and (b). The empirical miscoverage rate at any time $t$ is estimated by counting the fraction of runs at which the true state $\rho^*$ was not covered by the confidence set up to time $t$. The normalized set size is estimated by evaluating the fraction of the space of density matrices covered by the confidence set via sampling. A
uniform sampling ranging between 2000 and 5000 samples is used, following the procedure
detailed in~\cite[Sec. 2.2.2]{Blume_Kohout_2010}.

In both two-qubit and four-qubit settings, the set sizes for all schemes shrink as $t$ increases and more measurement data are collected.  Moreover, the empirical miscoverage rate for the proposed  AV-QST  remains below the required upper bound $\alpha$ specified in~(\ref{eq:prob2}) at all times, whereas the miscoverage rate of the B-QST credible region grows well above this bound. In contrast, LR-QST almost never exhibits miscoverage because it produces an overly large confidence region, which makes it substantially less informative than AV-QST. This result highlights a general conclusion from our experiments: While LR-QST cannot theoretically guarantee anytime coverage, in practice it is seen to be extremely conservative, yielding regions whose size is much larger than for AV-QST, especially for a larger number of qubits.

 Panels (c) in  Fig. \ref{fig:2qub_results} and Fig. \ref{fig:4qub_results} show the empirical miscoverage rate as a function of the target miscoverage rate $\alpha$ at $t=100$. B-QST sets are consistently smaller than the corresponding confidence sets produced by AV-QST (see panel (b)), but they fail to meet the desired miscoverage rate. In contrast, the LR-QST sets are always larger than the sets produced by AV-QST, confirming the over-conservativeness of LR-QST sets.

\section{Conclusion}  \label{sec:conclusion}

This paper has introduced anytime-valid confidence sets for quantum state  tomography. The theoretical properties of the  proposed approach, termed AV-QST, were seen via experiments to yield practical benefits in terms of coverage and set size. This work contributes to a recent line of research aiming at applying recent advances in statistics to quantum information processing \cite{park2023quantum,zecchin2025quantum}.   Among directions for future work, we mention the development of 
statistically valid change-detection mechanisms~\cite{shekhar2023reducingsequentialchangedetection, pmlr-v202-shekhar23a} and the design of adaptive measurement strategies, which are supported by the proposed framework.

\vspace{-0.2 cm}

\appendix[Proof of Proposition 1]
\vspace{-0.1 cm}

By defining $R_0(\rho^*)=1$, we can write~\eqref{eq:Martingale_R} via the recursive  relation
\begin{equation}
R_t(\rho^*) =  \frac{\mathrm{tr}(\Pi_{X_t} \hat{\rho}_{t-1})}{\mathrm{tr}(\Pi_{X_t} {\rho^*})} R_{t-1}(\rho^*).
\end{equation} Let $\mathscr{F}_t = \sigma(X_1, \dots, X_t)$ be the natural filtration. By the properties of conditional expectation, we then have
\begin{equation}\label{eq:der}
\begin{aligned}
    \mathbb{E}[R_t(\rho^*) \mid \mathscr{F}_{t-1}]
    &= R_{t-1}(\rho^*) \,
    \mathbb{E}\!\left[
        \frac{\mathrm{tr}(\Pi_{X_t} \hat{\rho}_{t-1})}
             {\mathrm{tr}(\Pi_{X_t} {\rho^*})} 
        \;\middle|\; \mathscr{F}_{t-1}
    \right]  \\
    &= R_{t-1}(\rho^*) 
       \sum_{x=1}^M \mathrm{tr}(\Pi_x {\rho^*})
       \frac{\mathrm{tr}(\Pi_x \hat{\rho}_{t-1})}
            {\mathrm{tr}(\Pi_x {\rho^*})} \\
    &= R_{t-1}(\rho^*).
\end{aligned}
\end{equation} The derivation (\ref{eq:der})  shows that the likelihood ratio sequence satisfies the martingale property. Furthermore, in a similar way, one can see that the expected value equals 1, i.e., $\mathbb{E}[R_t(\rho^*)]= \mathbb{E}[R_0(\rho^*)]=1$ for all times $t$. 
Thus the sequence $R_t(\rho^*)$ is a test martingale for the null hypothesis that the state is $\rho^*$. 

\bibliographystyle{IEEEtran}

\bibliography{bibtex/bib/IEEEexample}

\end{document}